\documentclass[%
 reprint,
 amsmath,amssymb,
 aps,
]{revtex4-1}

\usepackage{amsthm}
\usepackage{graphicx}
\usepackage{dcolumn}
\usepackage{bm}

\begin{document}

\preprint{APS/123-QED}

\title{Constructing locally indistinguishable orthogonal product bases in an $m \otimes n$ system}
\author{Guang-Bao Xu$^{1,2}$}

\author{Ying-Hui Yang$^{1,3}$}%

\author{Qiao-Yan Wen$^{1}$}

\author{Su-Juan Qin$^{1}$}

\author{Fei Gao$^{1}$}%
\email{gaofei\_bupt@hotmail.com}

\affiliation{%
 $^{1}$State Key Laboratory of Networking and Switching Technology, Beijing University of Posts and Telecommunications, Beijing, 100876, China\\
 $^{2}$College of Mathematics and Systems Science, Shandong University of Science and Technology, Qingdao, 266590, China\\
 $^{3}$School of Mathematics and Information Science, Henan Polytechnic University, Jiaozuo, 454000, China
}%

\date{\today}

\begin{abstract}
Recently, Zhang \emph{et al.} [Phys. Rev. A 92, 012332 (2015)] presented $4d-4$ orthogonal product states that are locally indistinguishable and completable in a $d\otimes d$ quantum system. Later, Zhang \emph{et al.} [arXiv: 1509.01814v2 (2015)] constructed  $2n-1$ orthogonal product states that are locally indistinguishable in $m\otimes n$ ($3\leq m \leq n$). In this paper, we construct a locally indistinguishable and completable orthogonal product basis with $4p-4$ members in a general $m\otimes n$ ($3\leq m \leq n$) quantum system, where $p$ is an arbitrary integer from $3$ to $m$, and  give a very simple but quite effective proof for its local indistinguishability.
Specially, we get a completable orthogonal product basis with $8$ members that cannot be locally distinguished in $m\otimes n$ ($3\leq m \leq n$)  when $p=3$. It is so far the smallest completable orthogonal product basis that cannot be locally distinguished in a $m\otimes n$ quantum system. On the other hand, we construct a small locally indistinguishable orthogonal product basis with $2p-1$ members, which is maybe uncompletable, in $m\otimes n$ ($3\leq m \leq n$ and $p$ is an arbitrary integer from $3$ to $m$). We also prove its local indistinguishability. As a corollary, we give an uncompletable orthogonal product basis with $5$ members that are locally indistinguishable in  $m\otimes n$ ($3\leq m \leq n$). All the results can lead us to a better understanding of the structure of a locally indistinguishable product basis in $m \otimes n$.
\begin{description}
\item[PACS numbers]
03.67.Dd, 03.67Mn, 03.67.Ac, 0365.Ud
\end{description}
\end{abstract}

\pacs{Valid PACS appear here}
\maketitle


\section{\label{sec:level1}Introduction\protect}
The local indistinguishability of orthogonal quantum states provides an effective tool to explore the relationship between quantum nonlocality and quantum entanglement. Bennett \emph{et al.} \cite{Bennett1999} firstly constructed a set of nine orthogonal product states that cannot be perfectly distinguished by local operations and classical communication (LOCC) in  $3\otimes 3$. Their work showed the counterintuitive phenomenon of nonlocality without entanglement, i.e., entanglement is not necessary for the local indistinguishability of orthogonal quantum states. Later, a simple proof for the nonlocality of the nine product states was given by Walgate \emph{et al.} \cite{Walgate2002}. Inspired by their work, many scholars are engaged in the research of this problem. With further research, numerous results \cite{CHB1999,Walgate2000,Chen2003,Niset2006,Jiang2010,Yu2011,Xin2008,Yang2013,Duan2010,Andrew2015,Chen2004,Rinaldis,MA2014,Feng2009} have been presented up to now. In spite of these huge advances, some basic problems are still incompletely solved, such as the smallest number of the states of an orthogonal product basis that can be completable and cannot be locally distinguished in a  high-dimensional bipartite quantum system.

To explore the local distinguishability of pure orthogonal product states in high-dimensional quantum system, Zhang \emph{et al.} \cite{Zhang2014} constructed $d^{2}$ orthogonal product basis quantum states that are locally indistinguishable in $d\otimes d$, where $d$ is odd and $d\geq 3$. Subsequently, Wang \emph{et al.} \cite{Wang2015} found a subset of the $d^{2}$ states, which contains $6d-9$ orthogonal states, is still locally indistinguishable. Later, Zhang \emph{et al.} \cite{Z.C.Zhang2015} constructed $4d-4$ orthogonal product states in $d\otimes d$ $(d\geq 3)$, and proved these states are locally indistinguishable. Yu \emph{et al.} \cite{Sixia2015} constructed $2d-1$ orthogonal states that are locally indistinguishable in $d\otimes d$ $(d\geq 3)$ and conjectured that any set of no more than $2(d-1)$ product states is locally distinguishable in $d\otimes d$ ($d\geq 3$). On the other hand, Wang \emph{et al.} \cite{Wang2015} presented a small set with only $3(m+n)-9$ orthogonal product states and proved that they cannot be perfectly distinguished by LOCC in the general $m\otimes n$ system, where $m\geq3$ and $n\geq 3$. Recently, Zhang \emph{et al.} \cite{Zhangzc2015} constructed $3n+m-4$ locally indistinguishable orthogonal product states that do not constitute a
unextendible product basis and presented a smaller set with $2n-1$ orthogonal product states that are LOCC indistinguishable in the general $m\otimes n$ ($3\leq m \leq n$) quantum system. All the results show it is a meaningful work to research the structure of the locally indistinguishable product basis and the smallest number of locally indistinguishable orthogonal product states in a high-dimensional quantum system.

In this paper, we construct a completable orthogonal product basis with $4p-4$ product states in $m\otimes n$ and give a simple proof for its local indistinguishability, where $3\leq m\leq n$ and $p$ is an arbitrary integer from $3$ to $m$. Specially, there exists a completable and locally indistinguishable orthogonal  basis with eight members in a $m\otimes n$ system ($3\leq m\leq n$). Eight is so far the
smallest number of locally indistinguishable states of a completable orthogonal product basis according to the existing results \cite{Zhangzc2015}. On the other hand, we construct a small orthogonal product basis that contains $2p-1$ orthogonal product states and is maybe uncompletable in $m\otimes n$ ,  and prove its local indistinguishability, where $3\leq m\leq n$ and $p$ is an arbitrary integer from $3$ to $m$.

\section{\label{sec:level1}PRELIMINARIES}

\theoremstyle{remark}
\newtheorem{definition}{\indent Definition}
\newtheorem{lemma}{\indent Lemma}
\newtheorem{theorem}{\indent Theorem}
\newtheorem{corollary}{\indent Corollary}
\def\QEDclosed{\mbox{\rule[0pt]{1.3ex}{1.3ex}}}
\def\QED{\QEDclosed}
\def\proof{\indent{\em Proof}.}
\def\endproof{\hspace*{\fill}~\QED\par\endtrivlist\unskip}

In this section, we introduce some definitions used throughout this paper.

\begin{definition} \cite{DiVincenzo2003}.
Consider a quantum system $H=\otimes_{i=1}^{q}H_{i}$ with $q$ parties. An orthogonal product basis (PB) is a set $S$ of pure orthogonal product states spanning a subspace $H_{S}$ of $H$. An uncompletable orthogonal product basis (UCPB) is a PB whose complementary subspace $H^{\perp}_{S}$ , \emph{i.e.}, the subspace in $H$ spanned by vectors that are orthogonal to all the vectors in $H_{S}$, contains fewer mutually orthogonal product states than its dimension. An unextendible product basis (UPB) is an uncompletable product basis for which $H^{\perp}_{S}$ contains no product state.
\end{definition}
  We call a PB is completable if it is not an
uncompletable orthogonal product basis.
\section{\label{sec:level1}Local indistinguishability of completable orthogonal product basis}

In this section, we construct a completable product basis with $4p-4$ members in $m\otimes n$ and prove its local indistinguishability, where $3\leq m\leq n$ and $p$ is an arbitrary integer from $3$ to $m$.
\begin{theorem}
In $m\otimes n$, the following $4p-4$ orthogonal product states,
\begin{eqnarray}
\nonumber
&&|\psi_{i}\rangle=\frac{1}{\sqrt{2}}|i\rangle_{A}|0-i\rangle_{B},\\
\nonumber
&&|\psi_{i+p-1}\rangle=\frac{1}{\sqrt{2}}|0-i\rangle_{A}|j\rangle_{B},\\
&&|\psi_{i+2p-2}\rangle=\frac{1}{\sqrt{2}}|i\rangle_{A}|0+i\rangle_{B},\\
\nonumber
&&|\psi_{i+3p-3}\rangle=\frac{1}{\sqrt{2}}|0+i\rangle_{A}|j\rangle_{B},\\
\nonumber
\end{eqnarray}
cannot be perfectly distinguished by LOCC, where $3\leq m\leq n$, $p$ is an arbitrary integer from $3$ to $m$, $j=i+1$ when $i=1,\cdots,\, p-2$ and $j=1$ while $i=p-1$. It is noted that $\frac{1}{\sqrt{2}}|\alpha\pm \beta\rangle=\frac{1}{\sqrt{2}}(|\alpha\rangle\pm|\beta\rangle)$
for $0\leq \alpha<\beta\leq p-1$ in this paper. 
\end{theorem}

\begin{proof}
Inspired by Ref. \cite{Sixia2015,Z.C.Zhang2015}, we give a very simple but quite effective method to prove the local indistinguishability of the $4p-4$ orthogonal product states. To distinguish these states, one of the two parties (Alice and Bob) has to start with a nondisturbing measurement, i.e., the postmeasurement states should be mutually orthogonal. Without loss of generality, suppose that Alice goes first with a set of general $m\times m$ positive operator-valued measure (POVM) elements $M_{t}^{\dag}M_{t}$ $(t=1,\cdots, l)$, where
\begin{equation}
\begin{split}
M_{t}^{\dag}M_{t}=
\left[
  \begin{array}{cccc}
    a_{00}^{t}&a_{01}^{t} &\cdots &a_{0(m-1)}^{t}\\
    a_{10}^{t}&a_{11}^{t} &\cdots &a_{1(m-1)}^{t}\\
    \vdots &\vdots  &\ddots &\vdots\\
    a_{(m-1)0}^{t} &a_{(m-1)1}^{t} &\cdots &a_{(m-1)(m-1)}^{t}
  \end{array}
\right]
\end{split}\nonumber
\end{equation}

The postmeasurement states $\{M_{t}\otimes I_{B}|\psi_{i}\rangle: i=1,\cdots, 4p-4\}$ should be mutually orthogonal.
  For the states $|\psi_{i}\rangle$
and
$|\psi_{j}\rangle$, where $1\leq i \leq p-1$, $1\leq j \leq p-1$ and $i \neq j$,  we have $$\langle i|M_{t}^{\dag}M_{t}|j\rangle\langle0-i|0-j\rangle=0,$$ Thus $\langle i|M_{t}^{\dag}M_{t}|j\rangle=0$, which means that $a_{ij}^{t}=0$ for $1\leq i \leq p-1$, $1\leq j \leq p-1$ and $i \neq j$.

For the states $|\psi_{j}\rangle$ and $|\psi_{i+p-1}\rangle$, where $j=i+1$ when $i=1,\cdots,p-2$ and $j=1$ when $i=p-1$, we can get $\langle j|M_{t}^{\dag}M_{t}|0-i\rangle\langle0-j|j\rangle=0$, Thus $\langle j|M_{t}^{\dag}M_{t}|0-i\rangle=\langle j|M_{t}^{\dag}M_{t}|0\rangle=0$. Similarly, we can get $\langle0|M_{t}^{\dag}M_{t}|j\rangle=0$. That means $a_{0j}^{t}=a_{j0}^{t}=0$ for $j=1,2,\cdots,p-1$.

For the states $|\psi_{i+p-1}\rangle=\frac{1}{\sqrt{2}}|0-i\rangle_{A}|j\rangle_{B}$ and $|\psi_{i+3p-3}\rangle=\frac{1}{\sqrt{2}}|0+i\rangle_{A}|j\rangle_{B}$, where $j=i+1$ when $i=1,\cdots, p-2$ and $j=1$ while $i=p-1$, we have $\langle0+i|M_{t}^{\dag}M_{t}|0-i\rangle\langle j|j\rangle=0$. Thus $\langle0|M_{t}^{\dag}M_{t}|0\rangle=\langle i|M_{t}^{\dag}M_{t}|i\rangle$. That is $a_{00}^{t}=a_{ii}^{t}$ for $i=1,\cdots,p-1$.

Therefore, we have
\begin{equation}
\begin{split}
M_{t}^{\dag}M_{t}=
\left[
  \begin{array}{cccccc}
    a_{00}^{t} &0 &\cdots &0 &a_{0p}^{t} &\cdots \\
    0 &a_{00}^{t} &\cdots &0  &a_{1p}^{t} &\cdots \\
    \vdots &\vdots &\ddots &\vdots &\vdots &\ddots \\
    0 &0 &\cdots &a_{00}^{t} &a_{2p}^{t} &\cdots \\
    a_{p0}^{t} &a_{p1}^{t} &\cdots &a_{p(p-1)}^{t} &a_{pp}^{t} &\cdots \\
    \vdots &\vdots &\ddots &\vdots &\vdots &\ddots \\
    \end{array}
\right]
\end{split}\nonumber
\end{equation}

Now we consider the probability of the measurement outcome corresponding to the measurement operator $M_{t}$ for each of the $4p-4$ states. It is easy to see
$$\langle\psi_{i}|(M_{t}^{\dag}M_{t})\otimes(I^{\dag}I)|\psi_{i}\rangle=a_{00}^{t}\,\,\,(\forall\,i\in\{1,\,2,\,\cdots,4p-4\})$$
for $t=1,2,\cdots, l$, where $\sum_{t=1}^{l}a_{00}^{t}=1$ according to the completeness of the measurement operators. This means any one of the $4p-4$ states can lead to the outcome that is corresponding to $M_{t}$ with the same probability $a_{00}^{t}$, i.e., the measurement $\{M_{t}\otimes I_{B}\}$ is trivial to the $4p-4$ states. In other words, Alice cannot get any information about which the measured state will be by the measurement $\{M_{t}\}$. Then if Bob starts with a nondisturbing measurement after Alice has done, he cannot get any useful information as Alice does because of the symmetry of the $4p-4$ states. That is, they cannot perfectly distinguish these states by LOCC.

On the other hand, if Bob goes first and then Alice does, these states cannot be perfectly distinguished either, which can be proved with the same method. So the $4p-4$ states in $m\otimes n$ cannot be perfectly distinguished by LOCC. This completes the proof.
\end{proof}

In fact, the $4p-4$ states of (1) are completable since they can become a completed orthogonal product basis in $m\otimes n$ by adding the following $mn-4p+4$ states:\\
$\{|0\rangle|0\rangle\}\cup\{|i\rangle|1\rangle\,|\,2\leq i \leq p-2\}\cup\{|i\rangle|2\rangle\, |\,3\leq i\leq p-1\}\cup\{|i\rangle|j\rangle\, |\,3\leq j \leq p-2; j+1\leq i\leq p-1 \,and\, 1\leq i \leq j-2\}\cup\{|i\rangle|p-1\rangle\,|\,1\leq i\leq p-3\}\cup\{|i\rangle|j\rangle\,|\,p\leq i\leq m-1,\, 0\leq j\leq n-1\}\cup\{|i\rangle|j\rangle\,|\,0\leq i\leq p-1,\, p\leq j\leq n-1\}$.

From Theorem 1, we know that the parameter $p$ can be an arbitrary integer from $3$ to $m$. Specially, we have the following corollary by Theorem 1 when $p=3$.
\begin{corollary}
In $m\otimes n$, the eight orthogonal product states,
\begin{eqnarray}
\nonumber
&&|\phi_{1,2}\rangle=\frac{1}{\sqrt{2}}|1\rangle_{A}|0\pm1\rangle_{B},\\
\nonumber
&&|\phi_{3,4}\rangle=\frac{1}{\sqrt{2}}|2\rangle_{A}|0\pm2\rangle_{B},\\
&&|\phi_{5,6}\rangle=\frac{1}{\sqrt{2}}|0\pm1\rangle_{A}|2\rangle_{B},\\
\nonumber
&&|\phi_{7,8}\rangle=\frac{1}{\sqrt{2}}|0\pm2\rangle_{A}|1\rangle_{B},
\end{eqnarray}
cannot be perfectly distinguished by LOCC, where $3\leq m\leq n$.
\end{corollary}
  From Refs. \cite{Z.C.Zhang2015,Zhangzc2015,Zhang2014,Wang2015,Sixia2015},
we know that the eight states is so far the smallest completable and locally indistinguishable orthogonal product basis in $m\otimes n$ ($3\leq m\leq n$). In fact, there exist two local unitary operators, say $U_{m\times m}$ and $V_{n\times n}$, such that
\begin{eqnarray}
\nonumber
&&U_{m\times m}|0\rangle_{A}=|1\rangle_{A},\,U_{m\times m}|1\rangle_{A}=|2\rangle_{A},\\
&&U_{m\times m}|2\rangle_{A}=|0\rangle_{A},\,\,V_{n\times n}|0\rangle_{B}=|1\rangle_{B},\\
\nonumber
&&V_{n\times n}|1\rangle_{B}=|2\rangle_{B},\,\,\,\,\,V_{n\times n}|2\rangle_{B}=|0\rangle_{B}.\\\nonumber
\end{eqnarray}
Since the local unitary operations do not change the local indistinguishability of orthogonal states, we have the following conclusion by Corollary 1.
\begin{corollary}
In $m\otimes n$, the eight orthogonal product states,
\begin{eqnarray}
\nonumber
&&|\phi'_{1,2}\rangle=\frac{1}{\sqrt{2}}|2\rangle_{A}|1\pm2\rangle_{B},\\
\nonumber
&&|\phi'_{3,4}\rangle=\frac{1}{\sqrt{2}}|0\rangle_{A}|0\pm1\rangle_{B},\\
&&|\phi'_{5,6}\rangle=\frac{1}{\sqrt{2}}|1\pm2\rangle_{A}|0\rangle_{B},\\
\nonumber
&&|\phi'_{7,8}\rangle=\frac{1}{\sqrt{2}}|0\pm1\rangle_{A}|2\rangle_{B}
\end{eqnarray}
cannot be perfectly distinguished by LOCC, where $3\leq m\leq n$.
\end{corollary}

Now, we give a simple proof for the local indistinguishability of the $d^{2}$ states in Ref. \cite{Zhang2014} and its subset with $6d-9$ states  in Ref. \cite{Wang2015} by Corollary 2, where $d$ is odd. In fact, both the $d^{2}$ states and its subset with $6d-9$ states in $d\otimes d$ contain the following eight states,
\begin{eqnarray}
\nonumber
&&|\varphi_{1,2}\rangle=\frac{1}{\sqrt{2}}|(d+3)/2\rangle_{A}[|(d+1)/2\rangle\pm|(d+3)/2\rangle]_{B},\\
\nonumber
&&|\varphi_{3,4}\rangle=\frac{1}{\sqrt{2}}|(d-1)/2\rangle_{A}[|(d-1)/2\rangle\pm|(d+1)/2\rangle]_{B},\\
&&|\varphi_{5,6}\rangle=\frac{1}{\sqrt{2}}[|(d+1)/2\rangle\pm|(d+3)/2\rangle]_{A}|(d-1)/2\rangle_{B},\,\,\\
\nonumber
&&|\varphi_{7,8}\rangle=\frac{1}{\sqrt{2}}[|(d-1)/2\rangle\pm|(d+1)/2\rangle]_{A}|(d+3)/2\rangle_{B}.\\
\nonumber
\end{eqnarray}
It is obvious that there exist two local unitary operators, say $U'_{d\times d}$ and $V'_{d\times d}$, which satisfy
\begin{eqnarray}
\nonumber
&&U'_{d\times d}|0\rangle_{A}=\frac{1}{\sqrt{2}}|(d-1)/2\rangle_{A},
U'_{d\times d}|1\rangle_{A}=\frac{1}{\sqrt{2}}|(d+1)/2\rangle_{A},\\
\nonumber
&&U'_{d\times d}|2\rangle_{A}=\frac{1}{\sqrt{2}}|(d+3)/2\rangle_{A},
V'_{d\times d}|0\rangle_{B}=\frac{1}{\sqrt{2}}|(d-1)/2\rangle_{B},\\
\nonumber
&&V'_{d\times d}|1\rangle_{B}=\frac{1}{\sqrt{2}}|(d+1)/2\rangle_{B},
V'_{d\times d}|2\rangle_{B}=\frac{1}{\sqrt{2}}|(d+3)/2\rangle_{B},
\nonumber
\end{eqnarray}
can change the eight states of (4) to the eight states of (5) when $m=n=d\geq 3$. Since the local unitary operations do not change the local indistinguishability of orthogonal quantum states, the eight states of (5) cannot be locally distinguished by Corollary 2. This means the local indistinguishability of the $d^{2}$ states \cite{Zhang2014} and its subset with $6d-9$ states \cite{Wang2015} can be directly got by Corollary 2.
\section{Local indistinguishability of small  orthogonal product basis}
In this section, we construct a small orthogonal product basis with $2p-1$ members in $m\otimes n$ ($3\leq m\leq n$, $p$ is an arbitrary integer from $3$ to $m$), which is maybe uncompletable, and prove its local indistinguishability.

\begin{theorem}
In $m\otimes n$, the following $2p-1$ orthogonal product states,
\begin{eqnarray}
\nonumber
&&|\psi_{i}\rangle=\frac{1}{\sqrt{2}}|i\rangle_{A}|0-i\rangle_{B},\\
&&|\psi_{i+p-1}\rangle=\frac{1}{\sqrt{2}}|0-i\rangle_{A}|j\rangle_{B},\\
\nonumber
&&|\psi_{2p-1}\rangle=\frac{1}{p}|0+1+\cdots+(p-1)\rangle_{A}|0+1+\cdots+(p-1)\rangle_{B}
\nonumber
\end{eqnarray}
cannot be perfectly distinguished by LOCC, where $3\leq m\leq n$, $p$ is an arbitrary integer from $3$ to $m$, $j=i+1$ when $i=1,\cdots, p-2$ and   $j=1$ while $i=p-1$.
\end{theorem}
\begin{proof} Similarly to the proof of Theorem 1, one of the two parties (Alice and Bob) has to start with a nondisturbing measurement to distinguish these states, i.e., the postmeasurement states should be mutually orthogonal. Without loss of generality, suppose that Alice goes first with a set of general $m\times m$ positive operator-valued measure (POVM) elements $M_{t}^{\dag}M_{t}$ $(t=1,\cdots, l)$, where
\begin{equation}
\begin{split}
M_{t}^{\dag}M_{t}=
\left[
  \begin{array}{cccc}
    a_{00}^{t}&a_{01}^{t} &\cdots &a_{0(m-1)}^{t}\\
    a_{10}^{t}&a_{11}^{t} &\cdots &a_{1(m-1)}^{t}\\
    \vdots &\vdots  &\ddots &\vdots\\
    a_{(m-1)0}^{t} &a_{(m-1)1}^{t} &\cdots &a_{(m-1)(m-1)}^{t}
  \end{array}
\right]
\end{split}\nonumber
\end{equation}

We can get $a_{ij}^{t}=0$ for $1\leq i \leq p-1$, $1\leq j \leq p-1$ and $i \neq j$, and $a_{0j}^{t}=a_{j0}^{t}=0$ for $j=1,2,\cdots,p-1$ by the same way as the proof of Theorem 1 since the postmeasurement states should be mutually orthogonal.

 For the states $|\psi_{i+p-1}\rangle=\frac{1}{\sqrt{2}}|0-i\rangle_{A}|j\rangle_{B}$
and $
|\psi_{2p-1}\rangle=\frac{1}{p}|0+1+\cdots+(p-1)\rangle_{A}|0+1+\cdots+(p-1)\rangle_{B}$, where $j=i+1$ when $i=1,\cdots, p-2$ and $j=1$ while $i=p-1$, we have
$\langle0-i|M_{t}^{\dag}M_{t}|0+\cdots+(p-1)\rangle\langle j|0+\cdots+(p-1)\rangle=\langle0-i|M_{t}^{\dag}M_{t}|0+\cdots+(p-1)\rangle=0$, i.e., $a_{00}^{t}=a_{ii}$ for $i=1,2,\cdots,p-1$.

 Therefore, we have
\begin{equation}
\begin{split}
M_{t}^{\dag}M_{t}=
\left[
  \begin{array}{cccccc}
    a_{00}^{t} &0 &\cdots &0 &a_{0p}^{t} &\cdots \\
    0 &a_{00}^{t} &\cdots &0  &a_{1p}^{t} &\cdots \\
    \vdots &\vdots &\ddots &\vdots &\vdots &\ddots \\
    0 &0 &\cdots &a_{00}^{t} &a_{2p}^{t} &\cdots \\
    a_{p0}^{t} &a_{p1}^{t} &\cdots &a_{p(p-1)}^{t} &a_{pp}^{t} &\cdots \\
    \vdots &\vdots &\ddots &\vdots &\vdots &\ddots \\
    \end{array}
\right]
\end{split}\nonumber
\end{equation}

Now we consider the probability of the measurement outcome corresponding to the measurement operator $M_{t}$ for each of the $2p-1$ states. It is easy to see
$$\langle\psi_{i}|(M_{t}^{\dag}M_{t})\otimes(I^{\dag}I)|\psi_{i}\rangle=a_{00}^{t}\,\,\,(\forall\,i\in\{1,\,2,\,\cdots,2p-1\})$$
for $t=1,2,\cdots, l$, where $\sum_{t=1}^{l}a_{00}^{t}=1$ according to the completeness of the measurement operators. This means any one of the $2p-1$ states can lead to the outcome that is corresponding to $M_{t}$ with the same probability $a_{00}^{t}$, i.e., the measurement $\{M_{t}\otimes I_{B}\}$ is trivial to the $2p-1$ states. That is, Alice cannot get any information about which the measured state will be by the measurement $\{M_{t}\}$. Then if Bob starts with a nondisturbing measurement after Alice has done, he cannot get any useful information as Alice does because of the symmetry of the $2p-1$ states. In fact, if Bob goes first and then Alice performs a nondisturbing measurement, they cannot distinguish the $2d-1$ states, either.

Therefore, the $2p-1$ states cannot be perfectly distinguished by LOCC. This completes the proof.
\end{proof}

By Theorem 2, the parameter $p$ can be an arbitrary   integer from $3$  to $m$. We have the following corollary directly according to Theorem 2 when $p=3$.

\begin{corollary}
In $m\otimes n$, the following five orthogonal product states,
\begin{eqnarray}
\nonumber
&&|\psi_{1}\rangle=\frac{1}{\sqrt{2}}|1\rangle_{A}|0-1\rangle_{B},\\
\nonumber
&&|\psi_{2}\rangle=\frac{1}{\sqrt{2}}|2\rangle_{A}|0-2\rangle_{B},\\
&&|\psi_{3}\rangle=\frac{1}{\sqrt{2}}|0-1\rangle_{A}|2\rangle_{B},\\
\nonumber
&&|\psi_{4}\rangle=\frac{1}{\sqrt{2}}|0-2\rangle_{A}|1\rangle_{B},\\
\nonumber
&&|\psi_{5}\rangle=\frac{1}{3}|0+1+2\rangle_{A}|0+1+2\rangle_{B}
\end{eqnarray}
are locally indistinguishable, where $3\leq m\leq n$.
\end{corollary}

Specially, the five states of (7) is a UPB when $m=n=3$. In fact, Bennett \emph{et al.} \cite{CHB1999} exhibits two results in [3]. One is that a UPB is not completable even in a locally extended Hilbert space. The other is that if a
set of orthogonal product states is exactly measurable by
LOCC, then the set can be completed in some extended
space. Thus, it is obvious that the five states of (7) are uncompletable and locally indistinguishable in $m\otimes n$ by the two results, which is coincident with Corollary 3. Since any $4$ orthogonal product states are shown to be locally distinguishable \cite{DiVincenzo2003}, it is easy to see that five is the smallest number of uncompletable and locally indistinguishable product states.
\section{Conclusion}
In this paper, we construct an completable orthogonal product basis with $4p-4$ members that cannot be perfectly distinguished by LOCC in $m\otimes n$, where $3\leq m\leq n$, $p$ is an arbitrary integer from $3$ to $m$, and give a simple but quite effective proof. As a special case, we get eight orthogonal product states that can be completable and cannot be locally distinguished in $m\otimes n$ ($3\leq m\leq n$). We show eight is so far the smallest number of locally indistinguishable and completable orthogonal product states.
We give a very simple proof for the local indistinguishability of the $d^{2}$ states in \cite{Zhang2014} and its subset with $6d-9$ states  in \cite{Wang2015}. On the other hand, we construct a samll locally indistinguishable orthogonal product basis with $2p-1$ members in $m\otimes n$, which is maybe uncompletable, where $3\leq m\leq n$ and $p$ is an arbitrary integer from $3$ to $m$. Our work will be useful for us to understand the structure both of completable and uncompletable product bases that cannot be distinguished by LOCC.

\begin{acknowledgments}
This work is supported by NSFC (Grant Nos. 61272057,
61572081, 61201431), Beijing Higher Education Young
Elite Teacher Project (Grant Nos. YETP0475, YETP0477).
\end{acknowledgments}

\nocite{*}
\bibliography{apssamp}

\begin{thebibliography}{}\label{sec:TeXbooks}
\bibitem{Bennett1999}
   C. H. Bennett, D. P.DiVincenzo, C. A. Fuchs, T.
   Mor, E. Rains, P. W. Shor, J. A. Smolin, and W. K.
   Wootters, Quantum nonlocality without entanglement, Phys. Rev. A 59, 1070(1999).
\bibitem{Walgate2002}
   J. Walgate and L. Hardy, Nonlocality, Asymmetry, and Distinguishing Bipartite States, Phys. Rev. Lett. 89, 147901 (2002).
\bibitem{CHB1999}
   C. H. Bennett, D. P. DiVincenzo, T. Mor, P. W.
   Shor, J. A. Smolin, and B. M. Terhal, Unextendible product bases and bound entanglement, Phys. Rev. Lett. 82, 5385 (1999).
\bibitem{Walgate2000}
   J. Walgate, A. J. Short, L. Hardy, and V. Vedral, Local distinguishability of multipartite orthogonal quantum states, Phys. Rev. Lett. 85, 4972 (2000).
\bibitem{Chen2003}
   P. X. Chen and C. Z. Li, Orthogonality and distinguishability: Criterion for local  distinguishability of arbitrary orthogonal states, Phys. Rev. A 68, 062107
   (2003).
\bibitem{Niset2006}J. Niset and N. J. Cerf,
   Multipartite nonlocality without entanglement in many dimensions, Phys. Rev. A 74, 052103 (2006).
\bibitem{Jiang2010}
   W. Jiang, X. J. Ren, Y. C. Wu, Z. W. Zhou, G. C.
   Guo and H. Fan, A sufficient and necessary condition for $2n-1$ orthogonal states to be locally distinguishable in a $C^{2}\otimes C^{n}$ system, J. Phys. A: Math. Theor. 43,
   325303 (2010).
\bibitem{Yu2011}
   N. K. Yu, R. Y. Duan and M. S. Ying, Any $2\otimes n$ subspace is locally distinguishable. Phys. Rev. A 84, 012304 (2011).
\bibitem{Yang2013}
   Y. H. Yang, F. Gao, G. J. Tian, T. Q. Cao and Q.
   Y. Wen, Local distinguishability of orthogonal quantum states in a $2\otimes2\otimes2$ system, Phys. Rev. A 88, 024301 (2013).
\bibitem{Xin2008}
   Y. Xin and R. Y. Duan, Local distinguishability of orthogonal $2\otimes3$ pure states. Phys. Rev. A 77, 012315 (2008).
\bibitem{Duan2010}
   R. Y. Duan, Y. Xin and M. S. Ying, Locally indistinguishable subspaces spanned by three-qubit unextendible product bases, Phys. Rev. A 81, 032329 (2010)
\bibitem{Andrew2015}
   Andrew M. Childs, Debbie Leung, Laura Man¡¦cinska, Maris Ozols. Commun. Math. Phys. 323, 1121¨C1153 (2013)
\bibitem{Chen2004}
   Chen, P. X.  Li, C. Z. Distinguishing the elements of a full product basis set needs only projective measurements and classical communication. Phys. Rev. A 70, 022306 (2004).
\bibitem{Rinaldis}
   S. De Rinaldis. Distinguishability of complete and unextendible product bases. Phys. Rev. A 70, 022309 (2004).
\bibitem{MA2014}
   T. Ma, M.-J. Zhao, Y.-K. Wang, and S.-M. Fei, Noncommutativity and local indistinguishability of quantum states, Sci. Rep. 4, 6336 (2014).
\bibitem{Feng2009}
   Y. Feng and Y. Y. Shi, Characterizing locally indistinguishable orthogonal product states,
   IEEE Trans. Inf. Theory 55, 2799 (2009).
\bibitem{Zhang2014}
   Z. C. Zhang, F. Gao, G. J. Tian, T. Q. Cao, and Q. Y. Wen, Nonlocality of orthogonal product basis quantum states, Phys. Rev. A 90, 022313 (2014).
\bibitem{Wang2015}
   Y. L. Wang, M. S. Li, Z. J. Zheng, and S. M. Fei, Nonlocality of orthogonal product-basis quantum states, Phys. Rev. A 92, 032313(2015).
\bibitem{Z.C.Zhang2015}
   Z. C. Zhang, F. Gao, S. J. Qin, Y. H. Yang, and Q. Y. Wen, Phys. Rev. A 92, 012332 (2015)
\bibitem{Sixia2015}
   S. X. Yu, and C.H. Oh, Detecting the local indistinguishability of maximally entangled states, arXiv: 1502.01274v1[quant-ph] (2015)
\bibitem{Zhangzc2015}
   Z. C. Zhang, F. Gao, Y. Cao, S. J. Qin, and Q. Y. Wen, Local indistinguishability of orthogonal product states, arXiv: 1509.01814v2 (2015)

\bibitem{DiVincenzo2003}
   D. P. DiVincenzo, T. Mor, P. W. Shor, J. A. Smolin, and B.M. Terhal, Unextendible product bases, uncompleteable product bases and bound entanglement, Commun. Math. Phys. 238, 379 (2003).
\end{thebibliography}
\end{document}